\documentclass[11pt]{article}
\usepackage{fullpage}
\usepackage{amsmath,amsthm,amsfonts,amssymb,graphicx,url}

\newcommand{\map}[3]{{{#1}:{#2}\rightarrow{#3}}}

\newcommand{\p}{\varphi}
\newcommand{\myand}{\;\mathrel{\&}\;}

\newcommand{\eps}{\varepsilon}

\newcommand{\cmpl}{\setminus}

\newcommand{\rotl}[2]{{\textup{RotL}_{{#1}}({#2})}}
\newcommand{\h}{\textup{halt}}

\theoremstyle{definition}
\newtheorem{definition}{Definition}

\theoremstyle{plain}
\newtheorem{theorem}[definition]{Theorem}
\newtheorem{lemma}[definition]{Lemma}

\newtheorem{corollary}[definition]{Corollary}
\newtheorem{claim}[definition]{Claim}

\newcommand{\halt}{{q_\h}}
\newcommand{\tuple}[1]{{\langle{#1}\rangle}}
\newcommand{\tmark}[2]{{[{#1},\downarrow\!{#2}]}}

\title{The complexity of some regex crossword problems}
\author{Stephen A. Fenner \\
University of South Carolina\thanks{Computer Science and Engineering Department, Columbia, SC 29208 USA\@.} \\
\texttt{fenner.sa@gmail.com}}

\begin{document}

\maketitle

\begin{abstract}
In a typical \emph{regular expression (regex) crossword} puzzle, you are given two nonempty lists $R_1,\ldots,R_m$ and $C_1,\ldots,C_n$ of regular expressions over some alphabet, and your goal is to fill in an $m\times n$ grid with letters from that alphabet so that the string formed by the $i$th row is in $L(R_i)$, and the string formed by the $j$th column is in $L(C_j)$, for all $1\le i\le m$ and $1\le j\le n$.  Such a grid is a \emph{solution} to the puzzle.  It is known that determining whether a solution exists is NP-complete.  We consider a number of restrictions and variants to this problem where all the $R_i$ are equal to some regular expression $R$, and all the $C_j$ are equal to some regular expression $C$.  We call the solution to such a puzzle an \emph{$(R,C)$-crossword}.  Our main results are the following:
\begin{enumerate}
\item
There exists a fixed regular expression $C$ over the alphabet $\{0,1\}$ such that the following problem is NP-complete: ``Given a regular expression $R$ over $\{0,1\}$ and positive integers $m$ and $n$ given in unary, does an $m\times n$ $(R,C)$-crossword exist?''  This improves the result mentioned above.
\item
The following problem is NP-hard: ``Given a regular expression $E$ over $\{0,1\}$ and positive integers $m$ and $n$ given in unary, does an $m\times n$ $(E,E)$-crossword exist?''
\item
There exists a fixed regular expression $C$ over $\{0,1\}$ such that the following problem is undecidable (equivalent to the Halting Problem): ``Given a regular expression $R$ over $\{0,1\}$, does an $(R,C)$-crossword exist (of any size)?''
\item
The following problem is undecidable (equivalent to the Halting Problem): ``Given a regular expression $E$ over $\{0,1\}$, does an $(E,E)$-crossword exist (of any size)?''
\end{enumerate}
\end{abstract}


\section{Introduction}

Regular expression crossword puzzles (regex crosswords, for short) share some traits in common with traditional crossword puzzles and with sudoku.  One is typically given two lists $R_1,\ldots,R_m$ and $C_1,\ldots,C_n$ of regular expressions labeling the rows and columns, respectively, of an $m\times n$ grid of blank squares.  The object is to fill in the squares with letters so that each row, read left to right as a string, \emph{matches} (i.e., is in the language denoted by) the corresponding regular expression, and similarly for each column, read top to bottom.  The solution itself may have some additional property, e.g., spelling out a phrase or sentence in row major order.
\begin{figure}
\begin{center}
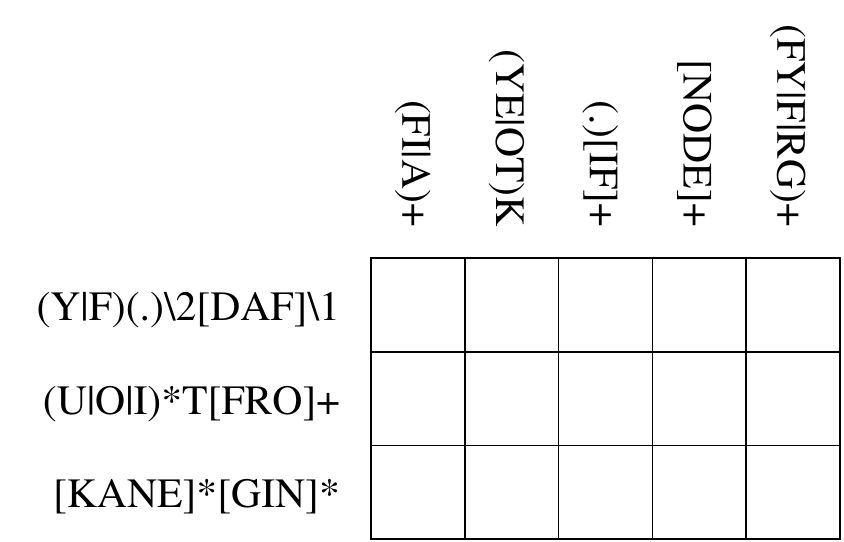
\end{center}
\caption{A regex crossword.  The alphabet is $\{A,B,\ldots,Z\}$.  The expression syntax includes: character classes $[\cdots]$, which are matched by any single letter between the brackets; the period ``.'', which matches any single letter; back references $\backslash 1$ and $\backslash 2$, which are matched by whatever string matched the first (respectively second) parenthesized subexpression.}\label{fig:royal-dinner}
\end{figure}
Figure~\ref{fig:royal-dinner} shows a $3\times 5$ regex crossword (``Royal Dinner'') from the website \url{regexcrossword.com}, with a unique solution \cite{royal-dinner}.

Regex crosswords have enjoyed some recent popularity, having been discussed in several popular media sources \cite{mikejuk:regex-crossword,Black:regex-crossword}, and thanks to some websites where people can solve the puzzles online \cite{rc,rcs}.  There are variants of the basic puzzle, including having two regular expressions for each row and column, one to match each of the two opposite directions \cite{rcs}.  Another variant is a hexagonal grid made up of hexagonal cells, with regular expressions for each of three separate directions, created by Dan Gulotta from an idea by Palmer Mebane (see \cite{Black:regex-crossword}), that showed up as part of the 2013 MIT Mystery Hunt \cite{MIT-Mystery-Hunt,Black:regex-crossword,mikejuk:regex-crossword}.

A natural complexity theoretic question to ask is: How hard is it to solve a regex crossword in general?  In September 2014, Glen Takahashi asked on StackExchange whether this task is NP-hard \cite{Takahashi:regex-crossword}.  (The same question has been asked by other people).  A positive answer to his query, along with a proof, was posted by FrankW about a half hour later (see Appendix~\ref{sec:NP-hard}, which includes a different proof found previously and independently by the author).  Another post observed the next day that the solution existence problem is in NP.

In this paper, we determine the complexity and decidability of several problems related to restricted regex crosswords.  In Section~\ref{sec:complexity}, we show that two restrictions of the problem above remain NP-hard: (1) when all the row expressions are equal to each other and all the column expressions are equal to a fixed expression, independent of the input; and (2) when all the expressions (both row and column) are equal.  In both problems above, the dimensions of the crossword are given in unary.  But first, in Section~\ref{sec:decidability}, we consider the problem of existence of crosswords of any size.  We show that this question and related ones are equivalent to the Halting Problem and are thus undecidable.  (We also show that if only one of the two dimensions is unbounded, then the problem becomes decidable, in fact, in PSPACE.)  Thus, in the spirit of the Post Correspondence Problem and questions about context-free languages, we have another simple yet undecidable problem in automata theory, one accessible to any undergraduate theory student.  The proofs given here are just as accessible, and their ideas carry over to the complexity theoretic setting of Section~\ref{sec:complexity}.

The undecidability results we prove in Section~\ref{sec:decidability} all follow from Lemma~\ref{lem:HP-to-RC}, which uses the crossword to encode a halting computation of a one-tape Turing machine as a two-dimensional tableau: one dimension for space, the other for time.  The row regular expression enforces consistency within each configuration, and the row and column regular expressions together enforce the legality of the machine's transitions.  In some sense, these are age-old techniques (albeit with a few new twists); Lemma~\ref{lem:HP-to-RC} could have been proved half a century ago.  They bear some similarity to results in cellular automata, to the Cook-Levin theorem, and to results of Berger from the 1960s showing the undecidability of tiling the plane with Wang tiles (the so-called ``domino problem'' \cite{Berger:domino-problem}, which was the first proof that there exist finite tile sets that tile the whole plane but only aperiodically).  Berger's construction is quite complicated, and although it may be possible to harness his result to prove our Lemma~\ref{lem:HP-to-RC}, our proof is direct enough to stand on its own.

The results of Section~\ref{sec:decidability} perhaps have their closest connection to the theory of two-dimensional languages (picture languages) in formal language theory \cite{GR:2d-languages}.  In fact, one can show that the recognizable picture languages coincide exactly with the letter-to-letter projections of $(R,C)$-crosswords \cite[Theorem~8.6]{GR:2d-languages} (except that the empty picture may also be included in the language).  Recognizable picture languages can be defined in terms of finite objects known as tiling systems \cite{GR:rec-picture-languages} (cf.\ \cite[Definition~7.2]{GR:2d-languages}), and given a tiling system $\mathcal{T}$, it is not hard to show that one can effectively find two regular expressions $R$ and $C$ (over some alphabet) and a projection $\pi$ that defines the same picture language as $\mathcal{T}$.  The existence problem for recognizable picture languages (``Given a tiling system, does it define a nonempty language?'') is known to be undecidable (\cite[Theorem~9.1]{GR:2d-languages}), and so, putting these results together, we get that the existence problem for $(R,C)$-crosswords is undecidable as well.  This essentially proves most of the results we give in Section~\ref{sec:decidability}, below, or at least weaker versions of them.  However, the proof we give for Lemma~\ref{lem:HP-to-RC} gives a much more direct reduction from the halting problem to $(R,C)$-crossword existence than what can be put together using the results in \cite{GR:2d-languages}.  In particular, we can fix the alphabet and even the expression $C$, independent of the input.  We will also need details of the proof later in the paper.


In the undecidability results of Section~\ref{sec:decidability} just described, the alphabet we use depends on the Turing machine being simulated and may be quite large if the machine recognizes the Halting Problem (e.g., a universal machine).  In Section~\ref{sec:binary}, which is the heart of the paper, we show that restricting to a binary alphabet does not reduce the complexity of any of our problems.  We do this by giving a polynomial-time function mapping a regular expression over any alphabet to one over a binary alphabet in a way that preserves the existence (and nonexistence) of crosswords (Lemma~\ref{lem:binary-alphabet}).  This turns out to be the most difficult result of the paper, using a surprisingly delicate construction that was arrived at only after many failed attempts.

We give a number of open problems in Section~\ref{sec:open}.

\section{Preliminaries}

For any integers $a$ and $b$ with $b > 0$, we define $a\bmod{b}$ to be the unique integer $r$ such that $0\le r < b$ and $a\equiv r\pmod{b}$.

Our notation for regular expressions is standard (see \cite{Sipser:theory3}), except that in addition to the three usual operations---union ($\cup$), concatenation (juxtaposition), and Kleene-closure ($*$)---we also allow intersection ($\cap$), so that $L(r\cap s) = L(r)\cap L(s)$ for any regular expressions $r$ and $s$.  We treat the $\cap$ operation only as syntactic sugar and not part of the formal definition of a regular expression; it can be effectively removed from a regular expression to obtain an equivalent one by standard techniques \cite{HMU:theory3,Sipser:theory3}.  We also use the unary ``$+$'' and ``$?$'' operators as syntactic sugar: ``$E^+$ is shorthand for $EE^*$, and ``$E?$'' is shorthand for $E\cup\eps$, where $\eps$ is matched by the empty string only.  We will also identify any finite set of strings with its corresponding regular expression.  We say a string $w$ \emph{matches} a regular expression $R$ to mean that $w\in L(R)$.

\begin{definition}
Let $\Sigma$ be an alphabet.  A \emph{$\Sigma$-crossword} is a nonempty two-dimensional array $X$ of symbols from $\Sigma$ (say, $m\times n$, where $m,n\ge 1$).  The symbol in the $i$th row and $j$th column is usually denoted $x_{i,j}$ for $1\le i\le m$ and $1\le j\le n$.  Let $R$ and $C$ be any regular expressions over $\Sigma$.  $X$ is an \emph{$(R,C)$-crossword}  iff each row of $X$, read left to right, matches $R$, and each column of $X$, read top to bottom, matches $C$; that is, for all $1\le i\le m$, the string $x_{i,1}x_{i,2}\cdots x_{i,n}$ is in $L(R)$ and for all $1\le j\le n$, the string $x_{1,j}x_{2,j}\cdots x_{m,j}$ is in $L(C)$.
\end{definition}

We may call a $\Sigma$-crossword simply a \emph{crossword} if $\Sigma$ is not relevant or is clear from the context.

The next definition is for purely technical reasons.  Removing these restrictions does not affect our complexity results.

\begin{definition}
We say that a regular expression is \emph{positive} iff it is not matched by the empty string.  A pair $(R,C)$ of regular expressions is \emph{plural} iff both $R$ and $C$ are positive and every $(R,C)$-crossword has at least two rows and at least two columns.
\end{definition}

Given two regular expressions $R$ and $C$, one can decide in polynomial time whether or not $R$ is positive and whether or not $(R,C)$ is plural.

We abbreviate ``computably enumerable'' (a.k.a.\ recursively enumerable) by ``c.e.''  Our notion of m-reduction (mapping reduction) and polynomial reduction come from Sipser \cite{Sipser:theory3}.

\section{Undecidability}
\label{sec:decidability}

In this section, we prove that, given a regular expression $R$, it is undecidable whether an $(R,C)$-crossword exists, for some fixed regular expression $C$.  We also show that it is undecidable whether an $(R,R)$-crossword exists.  As mentioned in the introduction, it is already known that, given \emph{both} regular expressions $R$ and $C$, determining whether an $(R,C)$-crossword exists is undecidable \cite{GR:2d-languages}.


We reduce from the Halting Problem.  Our computational model---a slight modification of that found in many textbooks, e.g., \cite{Sipser:theory3}---is that of a deterministic Turing machine with a unique halting state (distinct from the start state) and a single two-way infinite tape whose initial contents is an input string $w$ of nonblank symbols, surrounded on both sides with blank tape.  In each step, the tape head must move either left or right by one cell.  The crossword to be filled in encodes the tableau of a halting computation.  Each symbol in the crossword represents the contents of a tape cell at a certain time in the computation, possibly with some extra information about the state of the machine and the position of the head.  The expression $R$ ensures that the whole configuration of the TM is legitimate at each time step, and $C$ ensures that the contents of each tape cell is correct over time.  We view the tableau with the initial configuration on the top row and time moving downward.

One might think that, in order to handle transitions correctly, a crossword symbol should represent a ``window'' in the tableau, spanning perhaps two or three adjacent tape cells at two adjacent time steps, and that these windows should overlap consistently.  It is possible to do this, but it turns out to be unnecessary; we use a trick whereby the machine's transition information is passed in two directions---first horizontally (checked by $R$), then vertically (checked by $C$).  (This idea is somewhat analogous to the characterization of recognizable picture languages via domino systems and hv-local languages \cite{LS:domino-tiling}.)

Both results of this section use the following lemma:

\begin{lemma}\label{lem:HP-to-RC}
Let $M$ be a Turing machine (as described above).  There exists an alphabet $\Sigma$ and a regular expression $C := C(M)$ over $\Sigma$ (both depending on $M$), and for any input string $w$ there exists a regular expression $R := R(M,w)$ over $\Sigma$ (depending on $M$ and $w$) such that $(R,C)$ is plural, and $M$ halts on input $w$ if and only if an $(R,C)$-crossword exists, and if this is the case, then the $(R,C)$-crossword is unique.  Furthermore, $R$ is computable from $M$ and $w$ in polynomial time, and $C$ is computable from $M$.
\end{lemma}

\begin{proof}
Let $M = (Q,\Gamma,\delta,q_0,\halt,B)$, where
\begin{itemize}
\item
$Q$ is the (finite) state set,
\item
$q_0\in Q$ is the start state,
\item
$\halt\in Q$ is the halting state, different from $q_0$ ($M$ halts just when this state is entered),
\item
$\Gamma$ is the tape alphabet,
\item
$B\in\Gamma$ is the blank symbol, and
\item
$\map{\delta}{(Q\cmpl\{\halt\})\times\Gamma}{Q\times\Gamma\times\{L,R\}}$ is the transition function.  The left and right head directions are indicated by $L$ and $R$, respectively.
\end{itemize}

Given some input string $w\in(\Gamma\cmpl\{B\})^*$, we construct the two regular expressions $R$ and $C$ over an alphabet $\Sigma$ (defined below).  The expression $C$ only depends on $M$ and not on $w$.
For technical convenience and without loss of generality, we will make the following four additional assumptions about $M$'s computation: $M$'s head initially scans the blank cell immediately to the left of $w$; $M$'s initial transition is $\delta(q_0,B) = (q_1,B,L)$ for some state $q_1\ne q_0$; $M$ never re-enters state $q_0$ after its first step, nor scans any tape cells to the left of where it is after the first step (it might write a special symbol in the cell to keep itself from doing this); and at some point, $M$ scans the blank cell immediately to the \emph{right} of the input $w$ (which of course requires it to scan every symbol of $w$).  $M$ can be modified if necessary to meet these conditions without altering its halting \textit{versus} non-halting behavior on any input.

To avoid confusion, we will call the elements of the alphabet $\Sigma$ \emph{markers}, reserving the word \emph{symbol} to refer to elements of $\Gamma$.  The markers in $\Sigma$ are of the following three disjoint types:
\begin{description}
\item[Unscanned tape markers:]
For all $a\in\Gamma$, the marker $[a]$ is in $\Sigma$.  Each of these markers is used to depict a cell of the tape containing the symbol $a$ and which is scanned neither currently nor in the next time step.  We let $U := \{ [a] : a\in\Gamma\}$ denote the set of all unscanned tape markers.
\item[Scanned tape markers:]
For all $a\in\Gamma$ and all $q\in Q$, the marker $[a,q]$ is in $\Sigma$.  Each of these depicts a cell of the tape containing $a$ that is currently being scanned, and $M$'s current state is also included in the marker.
\item[State transmission markers:]
For all $a\in\Gamma$ and all $q\in Q\cmpl\{q_0\}$, the marker $\tmark{a}{q}$ is in $\Sigma$.  These markers depict tape cells that are currently unscanned but will be scanned in the next time step (and so they always appear horizontally adjacent to scanned tape markers for nonhalting states).  $M$'s state in the \emph{next} time step is also included in the marker.
\end{description}
To summarize: At each time step of $M$'s computation, the tape cell scanned by the head is recorded in the crossword by the scanned tape marker, which includes $M$'s current state.  All the unscanned cells of $M$'s tape are recorded in the crossword by their corresponding unscanned tape markers with one exception: the unscanned tape cell that will become scanned in the next time step will be recorded by a state transition marker, which includes $M$'s state in the next time step.

Here are two typical examples.  Suppose $M$'s current state is $q$ and it is scanning a $b$ on the tape, with $a$ to the left and $c$ to the right.  The corresponding configuration is traditionally denoted $\cdots aqbc\cdots$.  If $\delta(q,b) = (r,x,R)$, then the part of the crossword corresponding to the transition $aqbc \mapsto axrc$ looks like this:
\[ \begin{array}{|c|c|c|c|c|} \hline
\cdots & \cdots & \cdots & \cdots & \cdots \\ \hline
\cdots & [a] & [b,q] & \tmark{c}{r} & \cdots \\ \hline
\cdots & [a] & [x] & [c,r] & \cdots \\ \hline
\cdots & \cdots & \cdots & \cdots & \cdots \\ \hline
\end{array} \]
If instead, $\delta(q,b) = (s,y,L)$, then we get this for the transition $aqbc \mapsto sayc$:
\[ \begin{array}{|c|c|c|c|c|} \hline
\cdots & \cdots & \cdots & \cdots & \cdots \\ \hline
\cdots & \tmark{a}{s} & [b,q] & [c] & \cdots \\ \hline
\cdots & [a,s] & [y] & [c] & \cdots \\ \hline
\cdots & \cdots & \cdots & \cdots & \cdots \\ \hline
\end{array} \]

The one exception to this rule is a halting configuration, say $\cdots a\halt bc\cdots$, which is represented in the crossword thus:
\[ \begin{array}{|c|c|c|c|c|} \hline
\cdots & \cdots & \cdots & \cdots & \cdots \\ \hline
\cdots & [a] & [b,\halt] & [c] & \cdots \\ \hline
\end{array} \]
We will guarantee that there can be no rows of the crossword below this one.

\subsubsection*{The regular expression $R$}

$R$ ensures that all the rows of the crossword look like they should.  First we define a regular expression giving the initial configuration of $M$ on input $w$:  Let $w = w_1w_2\cdots w_n$, where $n\ge 0$ and each $w_i$ is in $\Gamma\cmpl\{B\}$.  Define
\begin{equation}\label{eqn:initial-tape}
I_w := \tmark{B}{q_1}[B,q_0][w_1][w_2]\cdots[w_n][B]^+\;.
\end{equation}
This is the only component of our construction that depends on the string $w$.  Since in its first step $M$'s head moves left and its state becomes $q_1$, this is the correct description of the first row.  Since $M$ never scans any cells further to the left thereafter, we can take $\tmark{B}{q_1}$ to start $I_w$.  Next, we define strings of markers indicating configurations beyond the initial one.  Set
\begin{align*}
T_L &:= \{ \tmark{b}{r}[a,q] : a,b\in\Gamma \myand q\in Q\cmpl\{q_0,\halt\} \myand (\exists c\in\Gamma)\delta(q,a) = (r,c,L) \}\;, \\
T_R &:= \{ [a,q]\tmark{b}{r} : a,b\in\Gamma \myand q\in Q\cmpl\{q_0,\halt\} \myand (\exists c\in\Gamma)\delta(q,a) = (r,c,R) \}\;, \\
T &:= T_L \cup T_R \cup \{[a,\halt] : a\in\Gamma \}\;,
\end{align*}
describing portions of the tape undergoing transitions.  Then we define
\[ R := I_w\cup U^*TU^*\;. \]
Note that $R$ requires each row to include exactly one scanned tape marker.  If the corresponding state is nonhalting, then it is adjacent to some state transmission marker (and this is the only place the latter marker can appear in the row).  If the corresponding state is halting, then there is no state transmission marker on the row.

Clearly, $R$ is positive and computable in polynomial time given $w$ and a description of $M$.

\subsubsection*{The regular expression $C$}

$C$ ensures that all the columns of the crossword look like they should.  We define $C := S\cap W$ as the intersection of two subexpressions: $S$ ensures that each tape cell stays constant (``static'')---except just after it is scanned by $M$'s head---and that when a cell becomes scanned, the new state information is faithfully copied from the previous time step; $W$ ensures that the correct symbol is written into a scanned cell on the next time step.

For $S$ we define
\begin{align*}
D &:= \bigcup_{a\in\Gamma,\,q\in Q\cmpl\{q_0\}}[a]^*\tmark{a}{q}[a,q]\;, \\
E &:= \bigcup_{a\in\Gamma,\,q\in Q\cmpl\{q_0\}}[a]^+\tmark{a}{q}[a,q]\;, \\
F &:= \bigcup_{a\in\Gamma}[a]^*\;, \\
S &:= (E \cup [B,q_0] \cup \tmark{B}{q_1}[B,q_1])D^*F\;.
\end{align*}
A string matching $E \cup [B,q_0] \cup \tmark{B}{q_1}[B,q_1]$ gives the contents of a tape cell starting at the beginning up through the first time it is scanned.  Thereafter, each string matching $D$ represents a time interval ending with the cell being scanned again.  $F$ is matched by the cell contents after the last time it is scanned.  Note that $S$ is positive, and hence $C$ is positive.

For $W$ we define
\begin{align*}
X &:= \{ [a,q][b] : a\in\Gamma \myand q\in Q\cmpl\{\halt\} \myand (\exists r\in Q)(\exists d\in\{L,R\})[\delta(a,q) = (r,b,d)] \} \\
Y &:= \{ [a,q]\tmark{b}{s} : a\in\Gamma \myand q\in Q\cmpl\{\halt\} \myand s\in Q\cmpl\{q_0\} \myand (\exists r\in Q)(\exists d\in\{L,R\})[\delta(a,q) = (r,b,d)] \} \\
H &:= \{ [a,\halt] : a\in\Gamma \} \\
Z &:= \Sigma \cmpl \{ [a,q] : a\in\Gamma \myand q\in Q \} \\
W &:= Z^*(XZ^*\cup Y)^*H?\;.
\end{align*}
Note that $W$ matches all strings in which any occurrence of a non-halting scanned tape marker is immediately followed by either an unscanned tape marker (or state transmission marker) giving the cell's correct contents after the corresponding transition of $M$.  $W$ also allows an optional halting scanned tape marker at the very end of the string.

Notice that $C$ is computable from $M$ alone and does not depend on the input string $w$ at all.  Note that we are \emph{not} asserting that $C$ is computable in polynomial time.  Our description of $C$ includes the intersection operator $\cap$, which is not part of the formal syntax of regular expressions.  As we mentioned, one can effectively compute an equivalent regular expression without the $\cap$ operator, but the resulting regular expression may be exponentially larger.

\subsubsection*{Correctness}

One direction of the lemma is now fairly clear: If $M$ halts starting with $w$ on its tape, then an $(R,C)$-crossword exists.  Such a crossword is also unique: $S$ makes sure that every column contains at least one scanned tape marker, and so the crossword represents exactly those tape cells that are scanned at least once by $M$ (which, by assumption, include the entire input string $w$); furthermore, any row containing a marker of the form $[a,\halt]$ (for some $a\in\Gamma$) must be the last row---this is enforced by $W$.  Finally, we note that, because $M$ makes at least one transition before it halts, the corresponding $(R,C)$-crossword has at least two rows and two columns, which makes $(R,C)$ plural.

For the other direction, suppose $X$ is an $(R,C)$-crossword.  Let $r_1,\ldots,r_m\in\Sigma^*$ and $c_1,\ldots,c_n\in\Sigma^*$ be the rows and columns of $X$, respectively, for some $m,n\ge 1$.  $S$ ensures that $r_1$ matches $(U\cup [B,q_0]\cup \tmark{B}{q_1})^*$, and since $R$ forces $r_1$ to contain a scanned tape marker somewhere, that marker must be $[B,q_0]$.  It follows that $r_1$ does not match $U^*TU^*$, and so it matches $I_w$, providing the right starting configuration for $M$ (and ensuring that $n\ge 2$).  We also have $m\ge 2$, ensured by $S$ because $r_1$ contains $\tmark{B}{q_1}$.  Thus $(R,C)$ is plural.  Subsequent rows must then conform to $M$'s computation, as was described previously.

Finally, the last row $r_m$ must contain a marker of the form $[a,\halt]$ for some $a\in\Gamma$, indicating that $M$ halts.  This is because $R$ ensures that $r_m$ contains some scanned tape marker, and supposing this marker is of the form $[a,q]$ for some $q\ne\halt$, there must be a state transmission marker on either side of it in $r_m$, whence $S$ ensures that this latter marker is followed by a scanned tape marker in its column, which means $r_m$ could not have been the last row.
\end{proof}

Lemma~\ref{lem:HP-to-RC} yields the following result:

\begin{theorem}\label{thm:HP-fixed-C}
There exists an alphabet $\Sigma$ and a positive regular expression $C$ over $\Sigma$ such that the decision problem
\begin{quote}
Given a regular expression $R$ over $\Sigma$ such that $(R,C)$ is plural, does an $(R,C)$-crossword exist?
\end{quote}
is m-equivalent to the Halting Problem (and is thus undecidable).
\end{theorem}

\begin{proof}
We apply Lemma~\ref{lem:HP-to-RC} letting $M$ be a universal Turing machine (or any Turing machine recognizing the Halting Problem).  Let $\Sigma$ and $C$ be as constructed in the proof.  Letting $W$ be the decision problem above, we get a computable function $g$ such that, for any string $w$, $g(w)$ is a regular expression $R$ such that $(R,C)$ is plural, and for all $w$, \ $M$ halts on $w$ if and only if an $(R,C)$-crosswords exists.  Thus $g$ m-reduces the Halting problem to $W$.  Conversely, $W$ is clearly c.e., and thus m-reduces to the Halting Problem.
\end{proof}

\begin{corollary}[Giammarresi, Restivo \cite{GR:2d-languages}]
Given regular expressions $R$ and $C$, it is undecidable (m-equivalent to the Halting Problem) whether an $(R,C)$-crossword exists.
\end{corollary}

\begin{proof}
Just note that $W$ in the proof of Theorem~\ref{thm:HP-fixed-C} is c.e.\ uniformly in $C$.
\end{proof}


%
%

\subsection{Making the row and column expressions equal}

The $(R,C)$-crossword existence problem remains undecidable even if we insist that $R = C$.  We get this from the following lemma:

\begin{lemma}\label{lem:HP-to-EE}
There exists a polynomial-time computable function $b$ such that, for any alphabet $\Sigma$ and any regular expressions $R$ and $C$ over $\Sigma$ such that $(R,C)$ is plural, $E := b(\Sigma,R,C)$ is a positive regular expression (over a slightly bigger alphabet $\Sigma'$) such that an $(E,E)$-crossword exists if and only if an $(R,C)$-crossword exists.  Furthermore, there is a one-to-one map $\rho$ mapping $\Sigma$-crosswords of size $m\times n$ (where $m,n\ge 2$) to $\Sigma'$-crosswords of size $(m+1)\times (n+1)$ that takes $(R,C)$-crosswords to $(E,E)$-crosswords, and for every $(E,E)$-crossword $Y$, there exists an $(R,C)$-crossword $X$ such that $\rho(X)$ is either $Y$ or the matrix transpose of $Y$.
\end{lemma}

\begin{proof}
Let $\Sigma$, $R$, and $C$ be given as in the lemma.  We want to effectively find an $E$ so that a unique $(E,E)$-crossword corresponds to any given $(R,C)$-crossword and vice versa.  A first attempt at constructing $E$ would be to set $E := R \cup C$.  This may not work, because an $(R,C)$-crossword may not exist, but there is an $(E,E)$-crossword where each row and column might match $R$, but the columns do not match $C$, say.  (In the case of $R$ and $C$ in the proof of Lemma~\ref{lem:HP-to-RC}, any square array with a single $[B,\halt]$ in each row and column, and the rest filled with all $[B]$'s is an $(R\cup C,R\cup C)$-crossword, regardless of $w$.)  There are perhaps several ways to correct this problem, and here is a fairly simple fix:
\begin{enumerate}
\item
Introduce three new symbols not in $\Sigma$: $\spadesuit$ (the ``bottom edge marker''); $\heartsuit$ (the ``left edge marker''); and $\diamondsuit$ (the ``corner marker'').
\item
Then modify $R$ and $C$ slightly to $R'$ and $C'$, respectively, so that any $(R'\cup C',R'\cup C')$-crossword or its matrix transpose has its first column matching $\heartsuit^*\diamondsuit$, its last row matching $\diamondsuit\spadesuit^*$, and the rest of the array being an $(R,C)$-crossword as before:
\[ \begin{array}{|c|ccc|} \hline
\heartsuit & & & \\
\vdots & & \mbox{$(R,C)$-crossword} & \\
\heartsuit & & & \\ \hline
\diamondsuit & \spadesuit & \cdots & \spadesuit \\ \hline
\end{array} \]
\end{enumerate}
Informally, the $\heartsuit$ and $\spadesuit$ markers prevent rows from being confused with columns, and the $\diamondsuit$ marker prevents $\heartsuit$ and $\spadesuit$ from being confused with each other.  Here are the formal definitions:
\begin{align*}
\Sigma' &:= \Sigma \cup \{\spadesuit,\heartsuit,\diamondsuit\}\;, \\
R' &:= \heartsuit R \cup \diamondsuit\spadesuit\spadesuit\spadesuit^*\;, \\
C' &:= C\spadesuit \cup \heartsuit\heartsuit\heartsuit^*\diamondsuit\;, \\
E &:= R' \cup C'\;.
\end{align*}
Clearly, $E = b(\Sigma,R,C)$ is positive and computable in polynomial time.  To see that this construction works, first observe that an $m\times n$ $(R,C)$-crossword $X$ (with $m,n\ge 2$ because $(R,C)$ is plural) becomes an $(m+1)\times(n+1)$ $(E,E)$-crossword $\rho(X)$ by prepending the column $\heartsuit^m$ then appending the row $\diamondsuit\spadesuit^n$.  This defines the map $\rho$, which is clearly one-to-one and maps $(R,C)$-crosswords to $(E,E)$-crosswords with one more row and column.  It follows that an $(E,E)$-crossword exists if an $(R,C)$-crossword exists.

Conversely, let $Y$ be any $(E,E)$-crossword---say, $m\times n$---with rows $r_1,\ldots,r_m$ and columns $c_1,\ldots,c_n$, all matching $E$.  We show first that $m,n\ge 3$.  Suppose not.  We must have $m,n\ge 2$, because both $R$ and $C$ are positive.  We may assume that $m=2$; otherwise, we apply the same argument to the transpose of $Y$, which is still an $(E,E)$-crossword.  Then each column of $Y$ has length $2$ and thus must match either $\heartsuit R$ or $C\spadesuit$.  Suppose $c_2$ starts with $\heartsuit$.  Then since $r_1$ has $\heartsuit$ as its second symbol, it must match $\heartsuit\heartsuit\heartsuit^*\diamondsuit$, whence $c_n$ starts with $\diamondsuit$; but then $|c_n|\ge 3$, contradicting our assumption that $m=2$.  Now suppose instead that $c_2$ matches $C\spadesuit$.  Then either $r_2$ matches $\diamondsuit\spadesuit\spadesuit\spadesuit^*$ or $r_2 = a\spadesuit$ for some $a\in\Sigma$ matching $C$.  The former case would make $c_1$ have $\diamondsuit$ as its second symbol, which is impossible.  In the latter case, we must have $c_1 = \heartsuit a$, which matches $\heartsuit R$.  But then $a$ matches both $R$ and $C$, making a $1\times 1$ $(R,C)$-crossword, which contradicts the fact that $(R,C)$ is plural.

Having established that $m,n\ge 3$, we next show that removing the first column and last row from either $Y$ or its transpose results in an $(R,C)$-crossword $X$, from which it will be clear that $\rho(X)$ is either $Y$ or its transpose, respectively.

Consider $r_2$, which has length $\ge 3$ and matches either $R'$ or $C'$.
\begin{description}
\item[Case 1:] $r_2$ matches $R'$.  Then $r_2$ must begin with $\heartsuit$: otherwise, it begins with $\diamondsuit$, but then $c_1$ has $\diamondsuit$ as its second symbol, which is impossible.  Then we have $r_2 = \heartsuit r$ for some string $r$ matching $R$, and since $c_1$ has $\heartsuit$ as its second symbol, we have $c_1 = \heartsuit^{m-1}\diamondsuit$, whence it follows that $r_m = \diamondsuit\spadesuit^{n-1}$.  Now consider the columns $c_2,\ldots,c_n$.  These all end with $\spadesuit$, and so they must all match $C\spadesuit$, because they all contain symbols in $\Sigma$ (from $r_2$).  So now we know that all symbols in $Y$ other than the first column and last row are in $\Sigma$, that is, for each $1\le i \le m-1$, all symbols in $r_i$, except possibly the first, are in $\Sigma$.  The only way this can happen is if each $r_i$ matches $\heartsuit R$.  This establishes that $Y$ minus the first column and last row is an $(R,C)$-crossword (whose image under $\rho$ is $Y$).
\item[Case 2:] $r_2$ matches $C'$.  By transposing $Y$, we can assume instead that $c_2$ matches $C'$, which is conceptually simpler.  The argument here is similar to Case~1.  The string $c_2$ cannot end with $\diamondsuit$, as that would also be the second symbol of $r_m$, which is impossible.  So we have that $c_2 = c\spadesuit$ for some string $c$ matching $C$, and thus $r_m = \diamondsuit\spadesuit^{n-1}$ (because $|r_m|\ge 3$), whence it follows that $c_1 = \heartsuit^{m-1}\diamondsuit$.  Now since $r_1,\ldots,r_{m-1}$ all start with $\heartsuit$ and contain at least one symbol from $\Sigma$, they all match $\heartsuit R$.  So again, all symbols in $Y$ except the first column and last row are from $\Sigma$, and since $c_2,\ldots,c_n$ all end in $\spadesuit$, they much all match $C\spadesuit$.  So again we have that deleting the first column and last row results in an $(R,C)$-crossword.
\end{description}
We have shown that removing the first column and last row from either $Y$ (in Case~1) or its transpose (in Case~2) results in an $(R,C)$-crossword $X$ such that $\rho(X)$ is either $Y$ or its transpose, respectively.  In particular, if an $(E,E)$-crossword exists, then an $(R,C)$-crossword exists.
\end{proof}

\begin{theorem}\label{thm:EE-is-undecidable}
Given a positive regular expression $E$, it is undecidable (in fact, m-equivalent to the Halting Problem) whether a $(E,E)$-crossword exists.
\end{theorem}

\begin{proof}
The problem is clearly c.e.\ and hence m-reduces to the Halting Problem.  Conversely, let $b$ be the function of Lemma~\ref{lem:HP-to-EE}, and let $M$, $\Sigma$, $C$, and $g$ be as in the proof of Theorem~\ref{thm:HP-fixed-C}.  Then $(g(w),C)$ is plural by Lemma~\ref{lem:HP-to-RC}.  For any string $w$, \ $M$ halts on $w$ if and only if a $(g(w),C)$-crossword exists.  Then, letting $E := b(\Sigma,g(w),C)$ (computable from $w$), we get by Lemma~\ref{lem:HP-to-EE} that $E$ is positive and that $M$ halts on $w$ if and only if an $(E,E)$-crossword exists.  Thus the mapping $b(\Sigma,g(\cdot),C)$ m-reduces the Halting problem to the $(E,E)$-crossword existence problem.
\end{proof}

\subsection{A decidable crossword existence problem}

In contrast with the previous results, we have the following:

\begin{theorem}\label{thm:decidable-RC}
There is an algorithm that decides, given a list of regular expressions $\tuple{R_1,\ldots,R_m}$ and a regular expression $C$ over an arbitrary alphabet $\Sigma$, whether there exists an $n\ge 1$ and an $m\times n$ array all of whose columns match $C$ and whose $i$th row matches $R_i$ for all $1\le i\le m$.  In fact, this decision problem is in PSPACE.
\end{theorem}

\begin{proof}[Proof Sketch]
First, we convert each $R_i$ into an equivalent $\epsilon$-NFA $N_i$ (see \cite{HMU:theory3}).  These automata have sizes polynomial in the sizes of the regular expressions.  Then we nondeterministically guess a crossword one column at a time, starting with the first, and for each guessed column, we simulate one step of each of the $N_i$ on its corresponding symbol (this can be done in polynomial time by keeping track of a subset of the state set of each $N_i$).  We accept if ever all the $N_i$ accept simultaneously.  We can also stop after $2^n$ guesses, where $n$ is the total number of states of all the $N_i$ combined.  This nondeterministic algorithm uses polynomial space, and hence can be converted into a deterministic polynomial-space algorithm by Savitch's theorem.
\end{proof}

\section{Regular expressions over the binary alphabet}
\label{sec:binary}

The alphabets used in Theorems~\ref{thm:HP-fixed-C} and \ref{thm:EE-is-undecidable} are fixed, but they are likely quite large, having to encode all the states of a universal Turing machine $M$.  In this section, we show how to map (in polynomial time) regular expressions over an arbitrary alphabet to regular expressions over the binary alphabet in a way that preserves crosswords.  Thus the crossword existence problem remains undecidable even when restricted to a binary alphabet.

\begin{lemma}\label{lem:binary-alphabet}
There is a function $f$ such that, for any $k\ge 2$ and positive regular expression $R$ over alphabet $\Sigma := \{0,\ldots,k-1\}$, \ $f(k,R)$ is a positive regular expression over the alphabet $\{0,1\}$ such that the following holds: There exists a one-to-one map $\psi_k$ between $\Sigma$-crosswords and $\{0,1\}$-crosswords (that maps $m\times n$ crosswords to $(3k(m+1)+1) \times  (3k(n+1)+1)$ crosswords) such that, for any positive regular expressions $T$ and $U$ over $\Sigma$,
\begin{enumerate}
\item
for any $(T,U)$-crossword $X$, \ $\psi_k(X)$ is a $(f(k,T),f(k,U))$-crossword, and
\item
for every $(f(k,T),f(k,U))$-crossword $Y$, there is a $(T,U)$-crossword $X$ such that $\psi_k(X) = Y$.
\end{enumerate}
Furthermore, $f$ is computable in time polynomial in $k + |R|$.
\end{lemma}

\begin{proof}
Fix $k$ and a positive regular expression $R$ over $\Sigma := \{0,\ldots,k-1\}$.
The regular expression $F := f(k,R)$ over $\{0,1\}$, defined below, will be formed from several components.  Let $\ell := 3k$, noting that $\ell \ge 6$.  Any string $w\in L(F)$ will satisfy $|w| \equiv 1 \pmod{\ell}$.  For $0\le i < \ell-1$ and any string $x$ of length $\ell$, define $\rotl{i}{x}$ to be the cyclic shift of $x$ by $i$ places to the left.  That is, if $x = x_0\cdots x_{\ell-1}$, then
\[ \rotl{i}{x} := x_i\cdots x_{\ell-1}x_0\cdots x_{i-1}\;. \]
Now define $s_0 := 0^{\ell-2}11$, and for $0 < i<\ell$ define $s_i := \rotl{i}{s_0}$.  We will use the $s_i$ to encode symbols from $\Sigma$.

Let $\map{h}{\Sigma^*}{\{0,1\}^*}$ be the string homomorphism determined by
\[ h(j) := s_{3j}\;, \]
for all $0\le j < k$.  We extend $h$ to apply to regular expressions over $\Sigma$ in the usual way (see \cite{HMU:theory3} for example).

Given a positive regular expression $R$ over $\Sigma$, the subexpressions making up $F := f(k,R)$ come in four types---alignment, calibration, encoding, and duplication---defined as follows:
\begin{description}
\item[Alignment:]
Define
\[ A := 1^\ell(0^\ell)^+\;. \]
\item[Calibration:]
Define
\begin{align*}
C_0 &:= 0001^{\ell-3}(s_0)^+\;, \\
C_1 &:= 01^{\ell-1}(s_1)^+\;, \\
C_2 &:= 01^{\ell-1}(s_2)^+\;,
\end{align*}
and for $3\le i < \ell-1$, define
\[ C_i := 1^\ell(s_i)^+\;. \]
Now define
\[ C := \bigcup_{i=0}^{\ell-1} C_i\;. \]
\item[Encoding:]
Define
\[ E^{(R)} := s_0(h(R))\;, \]
that is, $s_0$ concatenated with the regular expression $h(R)$.  Note that we make the dependence on $R$ explicit.  We use $E$ as shorthand for $E^{(\Sigma^+)}$ and note that $L(E^{(R)}) \subseteq L(E)$, because $R$ is positive.
\item[Duplication:]
Define
\[ D_0 := \bigcup_{1\le c<k} s_{3c}\;, \]
and for $j \in \{1,2\}$, define
\[ D_j := \bigcup_{0\le c<k} s_{3c+j}\;. \]
Define
\[ D :=  D_0(D_0)^+ \cup D_1(D_1)^+ \cup D_2(D_2)^+\;. \]
\end{description}
Finally, define
\[ F := 1(A \cup C) \cup 0(D \cup E^{(R)})\;. \]
This completes the description of $F = f(k,R)$.  It is evident that $f$ is computable in the specified time bounds.  Notice that all subexpressions of $F$ except $E^{(R)}$ depend only on $k$ and not on $R$.

Next we show how to convert any $\Sigma$-crossword $X$ into a unique $\{0,1\}$-crossword $Y = \psi_k(X)$ such that, for any positive regular expressions $T$ and $U$ over $\Sigma$, \ $X$ is a $(T,U)$-crossword if and only if $Y$ is an $(F,G)$-crossword, where $F := f(k,T)$ and $G := f(k,U)$.  It will help first to see an example of how this is done.  Suppose $\Sigma = \{0,1,2,3,4\}$.  Then each cell of a $\Sigma$-crossword is encoded by a $15\times 15$ square in the $\{0,1\}$-crossword, as shown in Figure~\ref{fig:encoded-letters}.
\begin{figure}
\begin{center}
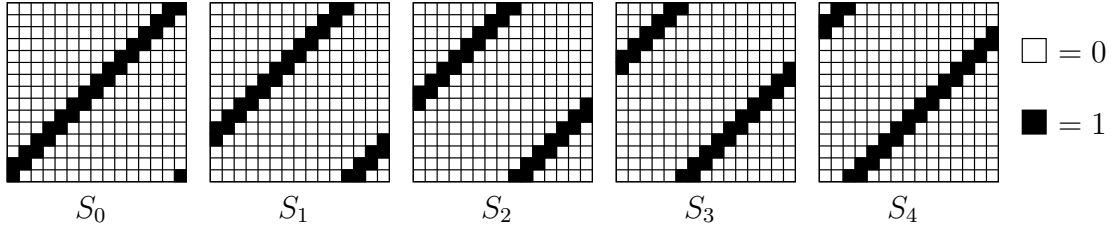
\caption{The $15\times 15$ squares $S_0,\ldots,S_4$ used to encode the individual letters $0,\ldots,4$, respectively.  A white cell denotes $0$, and a black cell denotes $1$.}\label{fig:encoded-letters}
\end{center}
\end{figure}
Generally, for $0\le c < k$ we define $S_c$ be the $\ell\times\ell$ square whose $i$th row (starting with $i=0$) is $s_{(3c+i)\bmod{\ell}}$.  These squares are pairwise distinct, and we use $S_c$ to encode the letter $c$.  Notice that the $S_c$ are symmetric (with respect to matrix transpose), and so the $i$th column of $S_c$ is also $s_{(3c+i)\bmod{\ell}}$.  In Figure~\ref{fig:sample-encoding},
\begin{figure}
\begin{center}
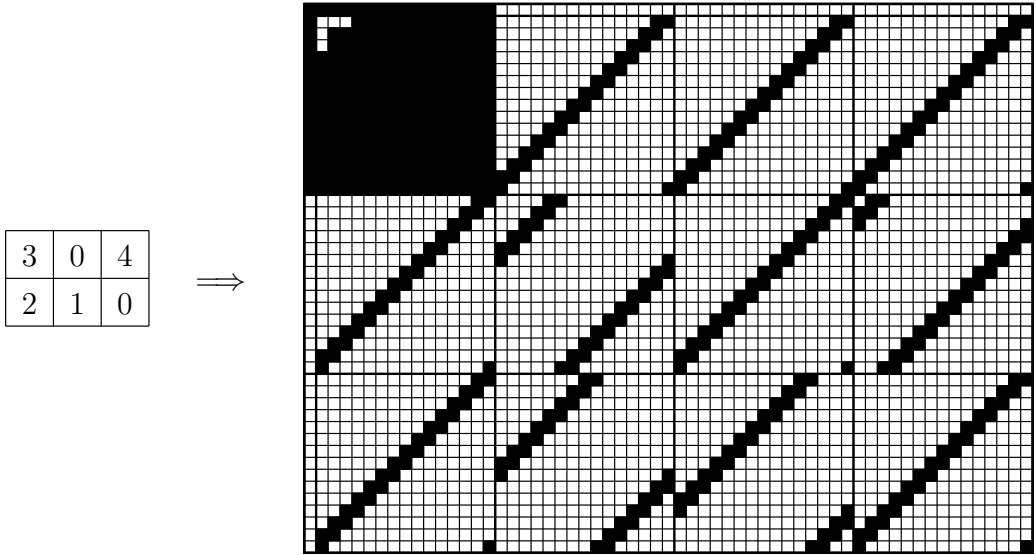
\end{center}
\caption{The encoding $\psi_5(X)$ of a sample $2\times 3$ $\Sigma$-crossword $X$, where $\Sigma = \{0,1,2,3,4\}$.  The slightly thicker lines give the boundaries between the $15\times 15$ squares.}\label{fig:sample-encoding}
\end{figure}
we show the encoding $\psi_k(X)$ of a sample $2\times 3$ $\Sigma$-crossword $X$.  The top row and left column form the \emph{alignment region}, and these two strings will both match $1A$.  The rest of the crossword is made up of $(\ell\times \ell)$-size squares $Q_{t,u}$ for $t,u\ge 0$, with $Q_{0,0}$ being the top leftmost square, $Q_{0,1}$ immediately to its right, $Q_{1,0}$ immediately below it, etc.  Squares of the form $Q_{0,u}$ and $Q_{t,0}$ form the \emph{calibration region}, and, except for $Q_{0,0}$, all these squares are equal to $S_0$.  The rows and columns making up this region all match $1C$.  The rest of the crossword (squares $Q_{t,u}$ for $t,u\ge 1$) forms the \emph{encoding region}, each square encoding a single corresponding entry in the $\Sigma$-crossword.  Rows and columns that intersect this region all match $0(D\cup E)$.

Now the detailed description.  Let $X$ be \emph{any} $\Sigma$-crossword with $m$ rows and $n$ columns, where $m,n\ge 1$.  For $1\le t \le m$ and $1\le u\le n$, let $x_{t,u}$ be the symbol in row~$t$ and column~$u$ of $X$.  Then we define a $\{0,1\}$-crossword $Y = \psi_k(X)$ as follows: $Y$ has dimensions $((m+1)\ell + 1)\times ((n+1)\ell + 1)$, where $\ell = 3k$ as above.  It will be convenient to index the rows of $Y$ as $(-1),\ldots,(m+1)\ell - 1$ and the columns as $(-1),\ldots,(n+1)\ell - 1$.  With this indexing, the alignment region comprises row~$(-1)$ and column~$(-1)$, and each square $Q_{t,u}$ (for $0\le t\le m$ and $0\le u\le n$) is the intersection of rows~$t\ell,\ldots,(t+1)\ell-1$ with columns $u\ell,\ldots,(u+1)\ell-1$.  We will define $Y$ row by row, with rows $r_{-1},\ldots,r_{(m+1)\ell-1}$, then discuss the columns.  (It will help to refer back to Figure~\ref{fig:sample-encoding}.)
\begin{itemize}
\item
Set $r_{-1} := 1^{\ell+1}0^{n\ell}$.  Then $r_{-1}$ matches $1A$.
\item
Set
\begin{align*}
r_0 &:= 10001^{\ell-3}(s_0)^n\;, \\
r_1 &:= 101^{\ell-1}(s_1)^n\;, \\
r_2 &:= 101^{\ell-1}(s_2)^n\;.
\end{align*}
Then $r_0$, $r_1$, and $r_2$ match $1C_0$, $1C_1$, and $1C_2$, respectively.
\item
For $3\le i < \ell$, set $r_i := 1^{\ell+1}(s_i)^n$.  Then $r_i$ matches $1C_i$.
\item
For $1\le t\le m$, let $x := x_{t,1}\cdots x_{t,n}$.  For $0\le i<\ell$, set
\[ r_{t\ell+i} := 0s_is_{(3x_{t,1}+i)\bmod{\ell}}\cdots s_{(3x_{t,n}+i)\bmod{\ell}}\;. \]
Note that for $1\le u\le n$, block~$u$ of $r_{t\ell+i}$ equals $\rotl{i}{h(x_{t,u})}$.  Also notice that if all the rows of $X$ match some positive regular expression $T$ over $\Sigma$, then all the $r_{t\ell}$ match $0E^{(T)}$.  The rest of the rows $r_{t\ell+i}$ match $0D$; in particular, $r_{t\ell+i}$ matches $0D_{i\bmod 3}$.
\end{itemize}
This completes the definition of the map $\psi_k$.

We have established that if the rows of $X$ all match some positive regular expression $T$, then each row of $Y$ matches $F = 1(A\cup C)\cup 0(D\cup E^{(T)})$, and from the arrangement of the rows, we can see by symmetry that if the columns of $X$ all match some positive regular expression $U$ over $\Sigma$, then each column of $Y$ matches $1(A\cup C)\cup 0(D \cup E^{(U)})$ in a similar manner:
\begin{itemize}
\item
$c_{-1} = 1^{\ell+1}0^{m\ell}$, matching $1A$.
\item
$c_0 = 10001^{\ell-3}(s_0)^m$, \ $c_1 = 101^{\ell-1}(s_1)^m$, and $c_2 = 101^{\ell-1}(s_2)^m$, matching $1C_0$, $1C_1$, and $1C_2$, respectively.
\item
For $3\le i < \ell$, \ $c_i = 1^{\ell+1}(s_i)^m$, matching $1C_i$.
\item
For $1\le u\le n$, letting $x := x_{1,u}\cdots x_{m,u}$, and for $0 \le i < \ell$, we have
\[ c_{u\ell + i} := 0s_is_{(3x_{1,u}+i)\bmod{\ell}}\cdots s_{(3x_{m,u}+i)\bmod{\ell}}\;. \]
That is, for $1\le t\le m$, block~$t$ of $c_{u\ell + i}$ equals 
Since $x$ matches $U$, we have that $c_{u\ell}$ matches $0E^{(U)}$, and the rest of the $c_{u\ell+i}$ match $0D$.
\end{itemize}
This establishes that, if $X$ is a $(T,U)$-crossword, then $Y = \psi_k(X)$ is an $(F,G)$-crossword (of the correct size), where $F := f(k,T)$ and $G = f(k,U)$.  It is also clear that, since $Y$ has the original crossword $X$ completely encoded within it, $\psi_k$ is a one-to-one map.

It remains to show that for any $(F,G)$-crossword $Y$, there is a $(T,U)$-crossword $X$ such that $\psi_k(X) = Y$, where $T$, $U$, $F$, and $G$ are as above.  We establish this through a series of claims.  Each claim is proved using ``sudoku-like'' arguments.  Let $Y$ be any $(F,G)$-crossword.  First observe that any string $w$ matching $A\cup C\cup D\cup E$ has length $v\ell$ for some $v\ge 2$, and so we can chop $w$ into substrings of length $\ell$ that we call \emph{blocks} (at least two), starting with block~$0$ through block~$v-1$.  This forces $Y$, minus its top row and left column, to be divided into $(\ell\times \ell)$-size squares $Q_{t,u}$ as described earlier, the rows and columns of each $Q_{t,u}$ being blocks in the rows and columns of $Y$ that intersect $Q_{t,u}$.  $Y$ has squares $Q_{t,u}$ for each $0\le t\le m$ and $0\le u\le n$ for some $m,n\ge 1$.  As before, we index the rows and columns of $Y$ as $-1,\ldots,(m+1)\ell-1$ and $-1,\ldots,(n+1)\ell-1$, respectively.

We extend the block concept to strings of length $v\ell+1$, e.g., the rows and columns of an $(F,G)$-crossword, by ignoring the first symbol in the string, that is, block~$0$ starts with the second symbol of the string.

\begin{claim}\label{claim:squares}
Each square $Q_{t,0}$ and $Q_{0,u}$ of $Y$, for $1\le t\le m$ and $1\le u\le n$, has exactly two $1$'s in each of its rows and each of its columns, the rest of the entries being $0$.
\end{claim}

\begin{proof}[Proof of Claim~\ref{claim:squares}]
Let $w$ be any string matching $A\cup C\cup D\cup E$.  Then each block of $w$, other than block~$0$, has at \emph{most} two $1$'s; in particular, it is either $0^\ell$ (if $w$ matches $A$), or it is of the form $s_i$ for some $i$ (if $w$ matches $C\cup D\cup E$).  Moreover, block~$0$ of $w$ has at \emph{least} two $1$'s.  Thus for $1\le u\le n$, square $Q_{0,u}$ has each of its rows containing at most two $1$'s and each of its columns containing at least two $1$'s.  The only way this can happen is if each row and column of $Q_{0,u}$ contains \emph{exactly} two $1$'s.  A similar argument shows that each row and column of $Q_{t,0}$ contains exactly two $1$'s, for $i\le t\le m$.
\end{proof}

\begin{claim}\label{claim:A}
No row other than the topmost, and no column other than the leftmost, matches $1A$.
\end{claim}

\begin{proof}[Proof of Claim~\ref{claim:A}]
Consider any row except the topmost.  This row is either $0r$ or $1r$ for some string $r$ matching $A\cup C\cup D\cup E$, and it intersects either $Q_{0,1}$ or else $Q_{t,0}$ for some $t\ge 1$.  In the former case, block~$1$ of $r$ (i.e., the block of $r$ intersecting $Q_{0,1}$) has a $1$, and so $r$ cannot match $A$; in the latter case, block~$0$ of $r$ has a $0$, and so again, $r$ cannot match $A$.  (Both cases follow from Claim~\ref{claim:squares}.)  Thus the row in question cannot match $1A$.  The same argument applies to the columns except the leftmost; none of them can match $1A$.
\end{proof}

\begin{claim}\label{claim:alignment}
The topmost row and leftmost column of $Y$ each match $1A$.
\end{claim}

\begin{proof}[Proof of Claim~\ref{claim:alignment}]
Observe that any string $w$ matching $C$ must have \emph{at least three} $1$'s in its block~$0$.  Now consider any row of $Y$ that intersects square $Q_{1,0}$.  This row is of the form $0r$ or $1r$, for some $r$ matching $A\cup C\cup D\cup E$.  By Claim~\ref{claim:A}, this row does not match $1A$, and so it must match $1C\cup 0(D\cup E)$.  However, $r$ cannot match $C$ because (by Claim~\ref{claim:squares}) $r$ has only two $1$'s in block~$0$.  Thus the row must match $0(D\cup E)$---in particular, it starts with $0$.  That means that the leftmost column (column~$(-1)$) has all $0$'s in its block~$1$, and so it cannot match $1C\cup 0(D\cup E)$, and thus it must match $1A$.  A similar, transposed argument shows that the topmost row must also match $1A$.
\end{proof}

\begin{claim}\label{claim:calibration}
Rows~$0,\ldots,\ell-1$ and columns $0,\ldots,\ell-1$ of $Y$ each match $1C$, and the rows and columns of $Y$ starting with index $\ell$ each match $0(D\cup E)$.
\end{claim}

\begin{proof}[Proof of Claim~\ref{claim:calibration}]
By the previous claim, the topmost row and leftmost column of $Y$ each match $1A = 1^{\ell+1}(0^\ell)^+$.  Thus rows $0,\ldots,\ell-1$ each start with $1$, and the rows starting with index $\ell$ each start with $0$.  By Claim~\ref{claim:A}, none of these rows match $1A$, so rows~$0$ through $\ell-1$ all must match $1C$ and the rest must match $0(D\cup E)$.  A similar argument holds for the columns.
\end{proof}

Notice that each row and column of $Q_{0,0}$ matches $(000\cup 011\cup 111)1^{\ell-3}$.  For $0\le i<(m+1)\ell$, let $r_i$ denote the row of $Y$ with index $i$, and for $0\le j < (n+1)\ell$ let $c_j$ denote the column of $Y$ with index $j$.  Rows $r_0,\ldots,r_{\ell-1}$ and columns $c_0,\ldots,c_{\ell-1}$ all match $1C$ by Claim~\ref{claim:calibration}, and the rest match $0(D\cup E)$.

\begin{claim}\label{claim:calibration-order}
For all $0\le i <\ell$, \ $r_i$ and $c_i$ both match $1C_i$.
\end{claim}

\begin{proof}[Proof of Claim~\ref{claim:calibration-order}]
First we show that $r_0$ and $c_0$ both match $1C_0$.  Suppose that $r_0$ does not match $1C_0$ (the argument for $c_0$ is similar).  Then (since $r_0$ matches $1C$) $r_0$ matches $1C_i$ for some $i\ge 1$, and so has a prefix matching $1(0\cup 1)1^{\ell-1}$, which makes $c_1,\ldots,c_{\ell-1}$ all have $11$ as a prefix.  This in turn implies that each of these columns must match $1C_j$ for some $j\ge 3$.  Now notice that block~$1$ of any string $x$ matching $C_j$ is $s_j$, and so if $3\le j<\ell$, then $x$ must have $0$ as the next to last symbol in its block~$1$.  From these facts it follows that the next to last row of $Q_{1,0}$ (i.e., block~$0$ of $r_{2\ell-2}$) matches $(0\cup 1)0^{\ell-1}$.  But this is impossible, because this block must have two $1$'s by Claim~\ref{claim:squares}.

Next we show that $r_1$ and $c_1$ match $1C_1$ and $r_2$ and $c_2$ match $1C_2$.  By what we just showed, $r_1$ $r_2$ both have prefix $10$ (because $c_0$ matches $1C_0$), and so they each match $1(C_0\cup C_1\cup C_2)$.  Neither of them can match $1C_0$, however: Consider the $2\times 2$ square $S$ forming the intersection or rows~$1,2$ with columns~$1,2$.  If either $r_1$ or $r_2$ matches $1C_0$, then $S$ contains a $0$, and hence at least one of the columns $c_1$ or $c_2$ must also match $1C_0$, which implies that $S$ contains \emph{all} $0$'s, which means that both $r_1$ and $r_2$ match $1C_0$.  But this would make $c_{2\ell-1}$ have prefix $0111$ putting three $1$'s in the last column of $Q_{0,1}$ and contradicting Claim~\ref{claim:squares}.  (By a similar argument, neither $c_1$ nor $c_2$ can match $1C_0$.)  Thus we have $r_1$ and $r_2$ both matching $1(C_1\cup C_2)$.  Now $r_2$ cannot match $1C_1$, for if it does, then $c_{2\ell-2}$ has prefix either $0101$ or $0111$, neither of which is possible because block~$0$ of $c_{2\ell-2}$ must be $s_j$ for some $j$.  Thus $r_2$ matches $1C_2$.  We have one more case to eliminate, i.e., showing that $r_1$ cannot match $1C_2$.  Suppose $r_1$ matches $1C_2$.  Then column $c_{2\ell-2}$ has prefix $0100$, and the only way this can happen is if $c_{2\ell-2}$ has prefix $0s_{\ell-1}$.  But that means that row $r_{\ell-1}$ has a $1$ as the next to last symbol of its block~$1$.  Since $r_{\ell-1}$ matches $1C$, this can only happen if $r_{\ell-1}$ matches $1(C_0\cup C_1)$, whence it has $10$ as a prefix.  This puts a $0$ as the last symbol of block~$0$ of $c_0$, but this is impossible, because $c_0$ matches $1C_0$ and hence has $10001^{\ell-3}$ as a prefix.  Thus $r_1$ cannot match $1C_2$, and so it matches $1C_1$.  A symmetric argument holds for $c_1$ and $c_2$.

Finally, we show that $r_i$ matches $1C_i$ for $3\le i < \ell$.  This is by induction on $i$, starting with $i=3$, with the inductive hypothesis that $r_j$ matches $1C_j$ for all $0\le j < i$.  We have then that $c_{2\ell-i-1}$ has prefix $0^i1$ and $c_{2\ell-i}$ has prefix $0^{i-1}11$.  Since both of these columns match $0(D\cup E)$ and hence must each start with $0s_j$ for some $j$'s, we can only have that $c_{2\ell-i-1}$ has prefix $0^i11$ and $c_{2\ell-i}$ has prefix $0^{i-1}110$.  Then block~$1$ of $r_i$ must be $s_i$, and it follows that $r_i$ matches $1C_i$.
\end{proof}

\begin{claim}\label{claim:cal-squares}
$Q_{t,0} = Q_{0,u} = S_0$ for all $1\le t \le m$ and $1\le u \le n$.
\end{claim}

\begin{proof}[Proof of Claim~\ref{claim:cal-squares}]
This follows immediately from Claim~\ref{claim:calibration-order}.
\end{proof}

\begin{claim}\label{claim:encoding-rows-columns}
For each $1\le t \le m$ and each $1\le u \le n$, \ $r_{t\ell}$ matches $0E^{(T)}$ and $c_{u\ell}$ matches $0E^{(U)}$.
\end{claim}

\begin{proof}[Proof of Claim~\ref{claim:encoding-rows-columns}]
By assumption, all rows of $Y$ match $F = 1(A\cup C)\cup 0(D\cup E^{(T)})$, and all columns of $Y$ match $G = 1(A\cup C)\cup 0(D\cup E^{(U)})$.  By Claim~\ref{claim:cal-squares}, each row $r_{t\ell}$ and each column $c_{u\ell}$ has prefix $0s_0$, and thus none can match $1(A\cup C) \cup 0D$.  Thus each such row must match $0E^{(T)}$, and each such column matches $0E^{(U)}$.
\end{proof}

\begin{claim}\label{claim:encoding-squares}
For all $t,u$ with $1\le t \le m$ and $1\le u \le n$, there exists a unique $x_{t,u}\in\Sigma$ such that $Q_{t,u} = S_{x_{t,u}}$.
\end{claim}

\begin{proof}[Proof of Claim~\ref{claim:encoding-squares}]
For simplicity, we will assume $t=u=1$; the same argument works for any $t,u$.  By Claim~\ref{claim:calibration}, rows $r_\ell,\ldots,r_{2\ell-1}$ and columns $c_\ell,\ldots,c_{2\ell-1}$ all match $0(D\cup E)$.  By Claim~\ref{claim:cal-squares}, the $i$th row of $Q_{1,0}$ (i.e., block~$0$ of $r_{\ell+i}$) is $s_i$, for $0\le i<\ell$.  We have $r_\ell$ matching $0E$ by Claim~\ref{claim:encoding-rows-columns}.  For $0\le j<\ell$, let $b_j$ be block~$1$ of $r_{\ell+j}$ (i.e., the $j$th row of $Q_{1,1}$), and let $b_j'$ be block~$1$ of $c_{\ell+j}$ (i.e., the $j$th column of $Q_{1,1}$).  Row $r_\ell$ matching $0E$ makes $b_0 = s_{3x}$ for some unique $0\le x < k$.  For $1\le i <\ell$, row $r_{\ell+i}$, having $s_i$ as its block~$0$, cannot match $0E$.  Thus $r_{\ell+i}$ matches $0D$, and in fact, it must match $0D_{i\bmod{3}}$, owing to its block~$0$, and this makes $b_i = s_{(3v+i)\bmod{\ell}}$ for some $0\le v<k$.  The same goes for the columns of $Q_{1,1}$.  Furthermore, notice that $D$ and $E$ ensure that the columns $b_j'$ and $b_{(j-1)\bmod{\ell}}'$ are distinct for any $0\le j<\ell$, because $j$ and $j-1$ have different remainders modulo $3$.

We show by induction on $1\le i<\ell$ that $b_i = s_{(3x+i)\bmod{\ell}}$, and this will imply that $Q_{1,1} = S_x$, finishing the proof of the claim.  Now assume (inductive hypothesis) that $b_{i-1} = s_{(3x+i-1)\bmod{\ell}}$ (we have established this for $i=1$).  We have $b_i = s_{3v+i}$ for some $0\le v<k$, and so it suffices to show that $v=x$.  Suppose $v\ne x$.  Then there is no position where the strings $b_i$ and $b_{i-1}$ share a $1$ in common.  The two $1$'s in $b_{i-1}$ occur in columns $b_{z_1}'$ and $b_{z_2}'$ of $Q_{1,1}$, where $z_1 := (-3x-i)\bmod{\ell}$ and $z_2 := (z_1-1)\bmod{\ell} = (-3x-i-1)\bmod{\ell}$, and by assumption, these $1$'s are then immediately followed by $0$'s in their respective columns.  Since $b_{z_1}' = s_{j_1}$ and $b_{z_2}' = s_{j_2}$ for some $0\le j_1,j_2<\ell$, and they share the substring $10$ in the same position in each, it must be that $j_1 = j_2$.  But this contradicts what we said above about columns being distinct.  Therefore, $v=x$, and we are done.
\end{proof}

\begin{claim}\label{claim:crossword-encoding}
For all $1\le t \le m$ and $1\le u \le n$, let $x_{t,u}\in\Sigma$ be the unique symbol such that $Q_{t,u} = S_{x_{t,u}}$ (cf.\ Claim~\ref{claim:encoding-squares}).  Then the $m\times n$ array $X$ whose $(t,u)$th entry is $x_{t,u}$ forms a $(T,U)$-crossword.
\end{claim}

\begin{proof}[Proof of Claim~\ref{claim:crossword-encoding}]
For $1\le t \le m$, let $d_t := x_{t,1}\cdots x_{t,n}$, and for $1\le u \le n$, let $e_u := x_{1,u}\cdots x_{m,u}$.  We show that the $d_t$ all match $T$ and the $e_u$ all match $U$.  We have
\[ r_{t\ell} = 0s_0s_{3x_{t,1}}\cdots s_{3x_{t,n}} = 0s_0(h(d_t))\;, \]
and because of the symmetry of the squares $Q_{t,u}$, we also have
\[ c_{u\ell} = 0s_0s_{3x_{1,u}}\cdots s_{3x_{m,u}} = 0s_0(h(e_u))\;, \]
for all $1\le t\le m$ and $1\le u\le n$.  By Claim~\ref{claim:encoding-rows-columns}, $r_{t\ell}$ matches $0E^{(T)} = 0s_0(h(T))$ and $c_{u\ell}$ matches $0E^{(U)} = 0s_0(h(U))$.  Then because $h$ is clearly a one-to-one map, it must be that $d_t$ matches $T$ and $e_u$ matches $u$.
\end{proof}

Finally, if $X$ is as defined in Claim~\ref{claim:crossword-encoding}, then is clear by our definition of $\psi_k$ above that $Y = \psi_k(X)$.  This ends the proof of Lemma~\ref{lem:binary-alphabet}.
\end{proof}

The next two theorems are just corollaries of Lemma~\ref{lem:binary-alphabet}.  They strengthen Theorems~\ref{thm:HP-fixed-C} and \ref{thm:EE-is-undecidable}, respectively.

\begin{theorem}\label{thm:HP-fixed-C-binary}
Let $G := f(k,C)$, where $C$ is as in Theorem~\ref{thm:HP-fixed-C}, and $k$ is the size of the alphabet used in that proof.  Then the following problem is m-equivalent to the Halting Problem:
\begin{quote}
Given a positive regular expression $F$ over the alphabet $\{0,1\}$, does an $(F,G)$-crossword exist?
\end{quote}
\end{theorem}

\begin{proof}
The problem is c.e.  For the other direction, we m-reduce from the problem of Theorem~\ref{thm:HP-fixed-C} via the map $f(k,\cdot)$.  Given any positive regular expression $R$ over an alphabet of size $k$, which we can assume is $\{0,\ldots,k-1\}$, we set $F := f(k,R)$.  Then an $(F,G)$-crossword exists if and only if an $(R,C)$-crossword exists, by Lemma~\ref{lem:binary-alphabet}.
\end{proof}

\begin{theorem}\label{thm:EE-binary}
The following problem is m-equivalent to the Halting Problem:
\begin{quote}
Given a positive regular expression $E'$ over the alphabet $\{0,1\}$, does an $(E',E')$-crossword exist?
\end{quote}
\end{theorem}

\begin{proof}
This works as in the proof of Theorem~\ref{thm:HP-fixed-C-binary}.  The problem is c.e.  Conversely, we m-reduce from the problem of Theorem~\ref{thm:EE-is-undecidable}.  Given a positive regular expression $E$, we can effectively determine the size $k$ of $E$'s alphabet.  Then adjusting the alphabet to $\{0,\ldots,k-1\}$, we let $E' := f(k,E)$, where $f$ is the function of Lemma~\ref{lem:binary-alphabet}.  Then $E'$ is positive, and an $(E',E')$-crossword exists if and only if an $(E,E)$-crossword exists.
\end{proof}

\section{Square crosswords}
\label{sec:square}

An $m\times n$ $\Sigma$-crossword is \emph{square} iff $m=n$.  In this section, we explain briefly why the complexities of all our problems are unaffected by restricting all crosswords to be square.

First, in the proof of Lemma~\ref{lem:HP-to-RC}, the $R$ and $C$ we construct are such that if an $m\times n$ $(R,C)$-crossword exists, then $m\ge n$.  This is because each row records a configuration of the machine $M$, and each column records a tape cell \emph{that is scanned at least once}, and $M$ can only scan at most as many different tape cells as there are configurations.  Thus to allow a square $(R,C)$-crossword, we only need to pad with (blank) cells that are never scanned.  Letting $C' := C\cup [B]^+$, we get that an $(R,C)$-crossword exists if and only if an $(R,C')$-crossword exists, if and only if a square $(R,C')$-crossword exists.

Next, the map $\rho$ of Lemma~\ref{lem:HP-to-EE} clearly preserves squareness: every $m\times n$ $(R,C)$-crossword (for $m,n\ge 2$) maps to an $(m+1)\times(n+1)$ $(E,E)$-crossword and vice versa.  Finally, the maps $\psi_k$ of Lemma~\ref{lem:binary-alphabet} also preserve squareness.  An $m\times n$ $(T,U)$-crossword maps under $\psi_k$ to a $(3k(m+1)+1)\times(3k(n+1)+1)$ $(F,G)$-crossword and vice versa.

\section{Complexity}
\label{sec:complexity}

It has been observed in \cite{Takahashi:regex-crossword}, and independently by us, that if separate regular expressions for each of the rows and columns are specified for a particular size grid, then the existence problem is NP-hard, and this is true even for a binary alphabet.  The proof (by FrankW) described in \cite{Takahashi:regex-crossword} is via a polynomial reduction from VERTEX~COVER, which, for the sake of completeness, we reproduce in Appendix~\ref{sec:NP-hard} as well as a reduction from 3-SAT that we found independently.  Both reductions map to regular expressions over binary alphabets.  In the former reduction, the regular expressions constructed for the columns, except for the first, are all the same fixed expression $0^*1(0\cup 1)^*$, independent of the input.  In our latter reduction, all the columnar regular expressions are the same fixed expression $0^+\cup 1^+$, independent of the input.  In each proof, however, the regular expressions for the rows are all different from each other.  Both results are therefore strengthened by Theorem~\ref{thm:NPC-fixed-C}, below, which is analogous to Theorem~\ref{thm:HP-fixed-C}.  First, a technical lemma.

\begin{lemma}\label{lem:SAT-to-RC}
There exist a polynomial $p$, a polynomial-time computable function $r$, and a positive regular expression $C'$ over $\Sigma$ such that, for any Boolean formula $\p$,
\begin{enumerate}
\item
$R' := r(\p)$ is a positive regular expression over $\{0,1\}$,
\item
$(R',C')$ is plural,
\item
every $(R',C')$-crossword is $q\times q$, where $q := q(|\p|)$, and
\item
the number of $(R',C')$-crosswords is equal to the number of satisfying truth assignments to $\p$.
\end{enumerate}
\end{lemma}

\begin{proof}
We modify slightly the proof of Lemma~\ref{lem:HP-to-RC} applied to a Turing machine $M$ such that, on any input $w$ of length $n$:
\begin{enumerate}
\item
$M$'s tape alphabet contains (at least) the nonblank symbols $0$ and $1$ and blank symbol $B$,
\item
$M$'s computation satisfies the technical conditions given at the start of that proof with respect to $w$,
\item
if $w$ encodes some Boolean formula $\p$ with variables $x_0,\ldots,x_{k-1}$ for some $k\le n$, then for any $a\in\{0,1\}^k$, with $wBa$ initially on its tape, $M$ scans $wBa$ in its entirety and halts if and only if $a$ is a satisfying truth assignment for $\p$, and
\item
if $M$ halts, then it halts after \emph{exactly} $p(n)-1$ many steps (thus including $p(n)$ many configurations), for some appropriately chosen polynomial $p$ with integer coefficients, independent of $w$, such that $p(n) \ge 2n+3$ for all $n\ge 0$.
\end{enumerate}
Such a machine $M$ and polynomial $p$ clearly exist.  Under these assumptions, we can change the definition of $I_w$ in Equation~\ref{eqn:initial-tape} to accommodate the presence of $a$ on the tape:
\[ I_w := \tmark{B}{q_1}[B,q_0][w_1]\cdots[w_n][B]([0]\cup[1])^k[B]^{p(n)-n-k-3}\;, \]
provided $w = w_1\cdots w_n$ encodes a Boolean formula with $k\le n$ variables.  Note that $I_w$ is only matched by strings of length $p(n)$.  The rest of the definition of $R$ remains the same.  We also modify $C$ just as we did in Section~\ref{sec:square}: $C'' := C \cup [B]^+$, where $C$ is as in the proof Lemma~\ref{lem:HP-to-RC}.  Under these modifications, both $R$ and $C''$ remain positive.  Now setting $p := p(n)$, we observe that for any $w$ encoding a Boolean formula $\p$ with $k\le n$ variables,
\begin{align*}
\mbox{$\p$ is satisfiable} &\iff \mbox{$M$ halts on $wBa$ for some $a\in\{0,1\}^k$} \\
&\iff \mbox{an $(R,C'')$-crossword exists,}
\end{align*}
and if such is the case, then owing to the determinism and running time of $M$, the $(R,C'')$-crossword is unique, is of size $p\times p$, and both $w$ and $a$ are easily recoverable from it, which implies that the number of $(R,C'')$-crosswords is equal to the number of satisfying assignments to $\p$.  Also by Lemma~\ref{lem:HP-to-RC}, given $\p$ we can compute $R$, $C$, and $0^p$ all in polynomial time.

Finally, we apply the function $f$ of Lemma~\ref{lem:binary-alphabet} to both $R$ and $C''$.  Let $\Sigma$ be the alphabet of $R$ and $C''$ (cf.\ Lemma~\ref{lem:HP-to-RC}).  By renaming if necessary, we may assume that $\Sigma = \{0,\ldots,\ell-1\}$ for some $\ell$.  Then we set
\begin{align*}
q &:= 3\ell(p+1)+1\;, \\
R' &:= f(\ell,R)\;, \\
C' &:= f(\ell,C'')\;, \\
r(\p) &:= R'\;.
\end{align*}
Any $(R',C')$-crossword thus has exactly $q = 3\ell(p+1)+1$ rows and columns.  The expressions $R'$ and $C'$ are both positive by Lemma~\ref{lem:binary-alphabet}, and so $(R',C')$ is plural, because $q\ge 2$.  Finally, since $f$ is polynomial-time computable (with constant $\ell$), so is $r$, and since $f$ preserves the number of crosswords, the number of $q\times q$ $(R',C')$-crosswords equals the number of $p\times p$ $(R,C'')$-crosswords, which equals the number of assignments satisfying $\p$.
\end{proof}

\begin{theorem}\label{thm:NPC-fixed-C}
For the regular expression $C'$ of Lemma~\ref{lem:SAT-to-RC}, the following decision problem $D$ is NP-complete (with respect to polynomial reductions):
\begin{quote}
Given as input a positive regular expression $R$ over the alphabet $\{0,1\}$ and a positive integer $q$ in unary, does an $q\times q$ $(R,C')$-crossword exist?
\end{quote}
\end{theorem}

\begin{proof}
$D$ belongs to the class NP, because one can verify in deterministic polynomial time whether or not a given $q\times q$ crossword (which has size polynomial in $q$) is an $(R,C')$-crossword.  (It is well-known that the problem, ``Given a regular expression $R$ and string $w$, does $w$ match $R$?'' is decidable in polynomial time, uniformly in $R$ and $w$.)

For NP-hardness, let $q := q(n)$ be the polynomial and $r$ the function defined in the proof of Lemma~\ref{lem:SAT-to-RC}.  Then that lemma implies that the map taking a Boolean formula $\p$ to $\tuple{r(\p),0^q}$ where $q := q(|\p|)$ is a polynomial reduction from SAT to $D$.
\end{proof}

The next theorem is analogous to Theorem~\ref{thm:EE-is-undecidable}.

\begin{theorem}
The following decision problem $S$ is NP-complete (with respect to polynomial reductions):
\begin{quote}
Given a positive regular expression $E$ over the alphabet $\{0,1\}$ and positive integer $s$ in unary, does an $s\times s$ $(E,E)$-crossword exist?
\end{quote}
\end{theorem}

\begin{proof}
$S$ is clearly in NP\@.  For NP-hardness, let $C'$, $q$, and $r$ be as in Lemma~\ref{lem:SAT-to-RC}, and let $b$ be the function of Lemma~\ref{lem:HP-to-EE}.  Given any Boolean formula $\p$, we show that we can produce in polynomial time a pair $\tuple{E,0^s}$ such that $\p$ is satisfiable if and only if an $s\times s$ $(E,E)$-crossword exists.  We do this by composing polynomial-time functions as follows:
\begin{enumerate}
\item
Let $R' := r(\p)$ and $q := q(|\p|)$.  Then, as in the proof of Theorem~\ref{thm:NPC-fixed-C}, a $q\times q$ $(R',C')$-crossword exists if and only if $\p$ is satisfiable.  Furthermore, $(R',C')$ is plural by Lemma~\ref{lem:SAT-to-RC}.
\item
Let $E' := b(\{0,1\},R',C')$.  Then by Lemma~\ref{lem:HP-to-EE}, $E'$ is a positive regular expression over the five-letter alphabet $\Sigma' := \{0,1,\heartsuit,\diamondsuit,\spadesuit\}$ such that a $(q+1)\times (q+1)$ $(E,E)$-crossword exists if and only if a $q\times q$ $(R',C')$-crossword exists.
\item
By renaming letters, we can assume that $\Sigma' = \{0,1,2,3,4\}$.  Then let $E := f(5,E')$, where $f$ is the function of Lemma~\ref{lem:binary-alphabet}, and let $s := 15(q+2)+1$.  Then by Lemma~\ref{lem:binary-alphabet}, $E$ is a positive regular expression over the alphabet $\{0,1\}$, and an $s\times s$ $(E,E)$-crossword exists if and only if a $(q+1)\times (q+1)$-crossword exists, if and only if a $q\times q$ $(R',C')$-crossword exists, if and only if $\p$ is satisfiable.
\end{enumerate}
The map $\p \mapsto \tuple{E,0^s}$ is thus a polynomial reduction from SAT to $S$.
\end{proof}

\begin{corollary}
For the regular expression $C'$ of Lemma~\ref{lem:SAT-to-RC}, the following two decision problems are NP-complete:
\begin{quote}
Given a positive regular expression $R$ over the alphabet $\{0,1\}$ and positive integers $m$ and $n$ in unary, does an $m\times n$ $(R,C')$-crossword exist?
\end{quote}
\begin{quote}
Given a positive regular expression $E$ over the alphabet $\{0,1\}$ and positive integers $m$ and $n$ in unary, does an $m\times n$ $(E,E)$-crossword exist?
\end{quote}
\end{corollary}

\subsection{Further results}

Since Lemma~\ref{lem:SAT-to-RC} controls not just the existence but the \emph{number} of crosswords, we can get more information out of it.  We list a few other results here that follow easily from Lemma~\ref{lem:SAT-to-RC}.
\begin{itemize}
\item
Counting the number of $(R,C)$-crosswords of given dimensions (given in unary) is polynomially equivalent to counting the number of satisfying assignments to a Boolean formula, and hence is complete for the class $\#P$ \cite{Valiant:NumP}.
\item
As with sudoku puzzles, someone who wants to solve a regex crossword puzzle (found online or in a newspaper, say) should reasonably expect that a solution exists and is unique.  Does the promise of a unique solution make solving the puzzle any easier in the worst case?  The answer is no, at least with respect to randomized polynomial reductions.  Consider the following search problem:
\begin{quote}
\underline{Input:} Regular expressions $R$ and $C$, and integers $m,n\ge 1$ in unary.\\
\underline{Promise:} A unique $m\times n$ $(R,C)$-crossword exists.\\
\underline{Ouput:} The $m\times n$ $(R,C)$-crossword.
\end{quote}
Lemma~\ref{lem:SAT-to-RC} and its proof says that this problem is polynomially equivalent to finding the unique satisfying assignment to a Boolean formula $\p$ with the promise that $\p$ is uniquely satisfiable.  This latter problem is known to be NP-hard with respect to randomized polynomial reductions \cite{VV}.
\item
Shifting perspective from the last item, a regex crossword puzzle \emph{maker} may want a test to determine, given regular expressions $R$ and $C$ and $m,n\ge 1$ in unary, whether or not a unique solution exists.  Lemma~\ref{lem:SAT-to-RC} says that this is polynomially equivalent to USAT, the language of all uniquely satisfiable Boolean formulas.  USAT is known to be NP-hard (it is in the class $\textup{D}^p$, the first level of the difference hierarchy over NP).
\end{itemize}

Finally, our techniques can be modified easily to show that if the dimensions of the crossword are both given in \emph{binary} instead of unary, then the $(R,C)$-crossword existence problem is complete for NEXP (nondeterministic exponential time) under polynomial reductions.  If one of the dimensions is given in unary and the other in binary, then the problem becomes PSPACE-complete.  (PSPACE-hardness follows from the techniques of Section~\ref{sec:complexity}; membership in PSPACE follows by modifying slightly the proof of Theorem~\ref{thm:decidable-RC}.)

\section{Open questions}
\label{sec:open}

We have shown that it is NP-hard to determine whether an $(R,C)$-crossword exists of given dimensions (specified in unary), even when $R$ and $C$ are over the binary alphabet.  Our reduction from SAT is rather complicated and indirect, however.  It would be nice to know if a simple, direct reduction from some NP-complete problem exists---perhaps some modification of one of the reductions given in Appendix~\ref{sec:NP-hard}.

Theorem~\ref{thm:HP-fixed-C} gives undecidability for a particular fixed expression $C$.  One may ask more generally: For which $C$ is the corresponding problem undecidable?  How hard is it to determine, given a $C$, whether the corresponding problem is decidable?  We conjecture that this latter question is m-complete for $\Sigma_3$, the third $\Sigma$-level of the arithmetic hierarchy (see, e.g., \cite{soare87}).  Similar questions can be asked about Theorem~\ref{thm:NPC-fixed-C}.  For example: For which $C$ is the question (i) NP-hard; (ii) decidable in polynomial time?

\subsection{Two-player regex crossword games}

One can imagine a variety of two-player games involving regex crosswords, and some of these may actually be fun to play.  For example:
\begin{enumerate}
\item
A blank $m\times n$ grid is given to start, along with regular expressions $R_1,\ldots,R_m$ and $C_1,\ldots,C_n$.  Player~1 (who plays rows) fills in the first row to match $R_1$, then Player~2 (who plays columns) fills in the rest of the first column so that it matches $C_1$, then Player~1 fills in the rest of row~2 so that it matches $R_2$, then Player~2 column~2, etc.
\item
Same as above, but each player can choose an incomplete row (respectively column) to fill in on each turn.
\item
Same as in item~1 above, but both players fill in rows in order, and a move is legal iff each column can be completed to match its corresponding $C_j$ (this may or may not be easy to determine).
\item
Same as in the last item, but players can choose rows to fill in on each turn.
\end{enumerate}
In all these games, the last player able to make a legal move wins.  We conjecture that for all these games, determining whether Player~1 has a winning strategy is PSPACE-hard, even if all the $R_i$ are equal and all the $C_j$ are equal and independent of the input, or if all the $R_i$ and $C_j$ are equal to each other.  (It is straightforward to prove that all these problems are in PSPACE.)

One might also consider some unbounded versions of these games:
\begin{enumerate}
\item
Positive regular expression $R$ and $C$ are given, but the size of the grid is not.  Player~1 first chooses an \emph{arbitrary} string $r_1$ matching $R$ for the first row of the grid (thus fixing the number of columns).  Player~2 then chooses an arbitrary string $c_1$ matching $C$ for the first column of the grid (except the first symbol of $c_1$ must equal that of $r_1$), thus fixing the number of rows.  Players then proceed as in the games mentioned previously.
\item
Same as the last item, but on their first move, each player chooses a string $r$ (respectively $c$) and says which row (respectively column) this string is to fill.
\end{enumerate}
The first two moves in each of these games is unbounded, but thereafter, the grid dimensions are fixed, and so determining the winner under optimal play is decidable, \emph{given the first two moves}.  The problem of determining if Player~1 wins \emph{without} knowing the first two moves is then in the class $\Sigma_2$, the second $\Sigma$-level of the arithmetic hierarchy (i.e., it is c.e.\ relative to the Halting Problem).  We conjecture that it is m-complete for this class.

\section*{Acknowledgments}

The author would like to thank Josh Cooper for introducing him to the subject by showing him the three-way regex crossword in \cite{Black:regex-crossword}.  He also thanks Jason O'Kane for first suggesting to him the NP-completeness question for regex crosswords as an exercise.  Finally, much of this work was done at the Dagstuhl seminar 14391, ``Algebra in Computational Complexity.''  The author wishes to thank the organizers of that seminar and especially Thomas Thierauf and the Technical University of Ulm (Germany) for their hospitality and lively discussions on this and other topics.  Thanks also to Klaus-J\"orn Lange for providing pointers to the literature on two-dimensional languages.

\bibliography{../bib/master}

\appendix

\section{Easy polynomial reductions from NP-complete problems}
\label{sec:NP-hard}

\begin{theorem}[FrankW \cite{Takahashi:regex-crossword}]
The following decision problem is NP-hard: ``Given lists of regular expressions $\tuple{R_1,\ldots,R_m}$ and $\tuple{C_1,\dots,C_n}$, all over the binary alphabet $\{0,1\}$, does there exist an $m\times n$ array of $0$'s and $1$'s whose $i$th row matches $R_i$ for all $1\le i\le m$ and whose $j$th column matches $C_j$ for all $1\le j\le n$?''
\end{theorem}

\begin{proof}
We describe a polynomial reduction from the NP-complete language VERTEX~COVER\@.  Let $\tuple{G,k}$ be an instance of VERTEX~COVER, where $G$ is a graph with $m$ vertices $v_1,\ldots,v_m$ and $n$ edges $e_1,\ldots,e_n$, and $k$ is a positive integer.  We define $C_1 := C_2 := \cdots := C_n := 0^*1(0\cup 1)^*$, which are matched by binary strings with at least one $1$.  We define $C_{n+1} := (0^*1?)^k0^*$, which is matched by binary strings with at most $k$ many $1$'s.  For $1\le i\le m$, let $r_i$ be the length $n$ string whose $j$th symbol is $1$ iff $v_i$ is an endpoint of edge $e_j$, and $0$ otherwise, then define $R_i := r_i1\cup 0^*$.  This construction can clearly be done in polynomial time.  Then we show that an $m\times(n+1)$ crossword exists where the $i$th row matches $R_i$ and the $j$th column matches $C_j$, for all $i,j$, if and only if $G$ has a vertex cover of size $\le k$.

To see this, first assume that $C$ is a vertex cover for $G$ of size $\le k$.  Then let the $i$th row be $r_i1$ if $v_i\in C$ and $0^{n+1}$ otherwise.  Then there is at least one $1$ in each of the first $n$ columns, because at least one endpoint of each edge is in $C$.  Also, there are at most $k$ many $1$'s in the $(n+1)$st column, because $|C|\le k$.  Thus all the row and column regular expressions are matched.  Conversely, suppose all the row and column expressions are matched.  Let $C := \{ v_i \mid \mbox{row~$i$ ends with $1$}\}$.  Then $|C|\le k$ due to the last column, and the $i$th row must be $r_i1$ for all $i$ such that $v_i\in C$.  Then $C$ is a vertex cover, because each of the first $n$ columns contains a $1$ and hence each edge has an endpoint in $C$.
\end{proof}

\begin{theorem}
The following decision problem is NP-complete: ``Given a list of regular expressions $\tuple{R_1,\ldots,R_m}$, all over the binary alphabet $\{0,1\}$, and a positive integer $n$ in unary, does there exist an $m\times n$ array of $0$'s and $1$'s, all of whose columns match $0^*\cup 1^*$ and whose $i$th row matches $R_i$ for all $1\le i\le m$?''
\end{theorem}

\begin{proof}
The problem is in NP because testing whether a given string $w$ matches a given regular expression $E$ can be done in polynomial time, uniformly in $|\tuple{w,E}|$.  To show NP-hardness, we reduce from $3$-SAT\@.  Given a $3$-cnf Boolean formula $\p = c_1 \wedge \cdots\wedge c_m$ over Boolean variables $x_1,\ldots,x_n$, we construct an instance $\tuple{R_1,\ldots,R_m}$ so that any $0$--$1$ array satisfying the criterion is $m\times n$ and encodes the truth value of each $x_j$ with respect to some satisfying assignment in the $j$th-column.  Each $R_i$ ensures that the $i$th clause is satisfied by the assignment.

If the $j$th column matches $C := 0^*\cup 1^*$, then it is either all $1$s, meaning $x_j$ is set to TRUE, or all $0$s, meaning $x_j$ is set to FALSE\@.  For $1\le i\le m$, suppose $c_i = \ell_{i_1} \vee \ell_{i_2} \vee \ell_{i_3}$, where $1\le i_1 < i_2 < i_3 \le n$ and each literal $\ell_k$ is either $x_k$ or $\overline{x_k}$.  For $j=1,2,3$, set $b_j := 1$ if $\ell_{i_j} = x_{i_j}$ and $b_j := 0$ otherwise.  Then finally set
\[ R_i = (0\cup 1)^{i_1-1}b_1(0\cup 1)^{n-i_1} \cup (0\cup 1)^{i_2-1}b_2(0\cup 1)^{n-i_2} \cup (0\cup 1)^{i_3-1}b_3(0\cup 1)^{n-i_3} \;. \]
This construction can clearly be done in polynomial time.

$R_i$ is matched by any string that contains, at position $i_1$ or $i_2$ or $i_3$, a truth value satisfying the corresponding literal.  $C$ guarantees that the truth values are consistent across all clauses.  Thus such an array exists if and only if $\p$ is satisfiable.
\end{proof}

\end{document}